\documentclass[conference]{IEEEtran}
\IEEEoverridecommandlockouts
\usepackage{cite}
\usepackage{amsmath,amssymb,amsfonts}
\usepackage{graphicx}
\usepackage{textcomp}
\usepackage{xcolor}
\def\BibTeX{{\rm B\kern-.05em{\sc i\kern-.025em b}\kern-.08em
    T\kern-.1667em\lower.7ex\hbox{E}\kern-.125emX}}
\usepackage{url}

\usepackage{cite}
\usepackage{amsmath, amssymb, amsfonts, amsthm}
\usepackage{algorithm}
\usepackage[noend]{algpseudocode} 
\usepackage{graphicx}
\usepackage{textcomp}
\usepackage{xcolor}
\usepackage{listings}
\usepackage{braket}
\usepackage{physics}
\RequirePackage{xspace}
\usepackage[colorlinks=true, allcolors=blue]{hyperref}
\algrenewcommand\algorithmicrequire{\textbf{Input:}}
\algrenewcommand\algorithmicensure{\textbf{Output:}}

\DeclareMathOperator*{\argmin}{arg\,min}

\theoremstyle{plain}
\newtheorem{theorem}{Theorem}

\newtheorem{proposition}[theorem]{Proposition}
\newtheorem{conjecture}[theorem]{Conjecture}

\newtheorem*{theorem*}{Theorem}
\newtheorem*{lemma*}{Lemma}
\newtheorem*{proposition*}{Proposition}
\newtheorem*{conjecture*}{Conjecture}
\newtheorem{fact*}{Fact}
\newtheorem*{axiom}{Axiom}

\theoremstyle{definition}

\newtheorem*{definition*}{Definition}
\newtheorem*{question*}{Question}
\newtheorem*{example*}{Example}
\newtheorem*{remark*}{Remark}
\newtheorem*{remarks*}{Remarks}
\newtheorem*{exercise*}{Exercise}
\newtheorem*{assumption*}{Assumption}

\newcommand{\approxd}{\stackrel{\rm d}{\approx}}
\newcommand{\approxst}{\approx_{\rm st}}

\newcommand{\pr}{\mathbb{P}}

\newcommand{\floor}[1]{\left\lfloor #1 \right\rfloor}

\def\BibTeX{{\rm B\kern-.05em{\sc i\kern-.025em b}\kern-.08em
    T\kern-.1667em\lower.7ex\hbox{E}\kern-.125emX}}
\begin{document}

\title{Towards analyzing large graphs with quantum annealing and quantum gate computers}

\author{

\IEEEauthorblockN{ 1\textsuperscript{st} Hannu Reittu* \thanks{* corresponding author}}
\IEEEauthorblockA{\textit{ Collective and collaborative AI  } \\
\textit{VTT Technical Research Centre of Finland}\\
P.O. Box 1000, FI-02044 VTT, Finland \\
hannu.reittu@vtt.fi}\\

\IEEEauthorblockN{ 3\textsuperscript{rd} Lasse Leskel\"a}
\IEEEauthorblockA{\textit{Dept.~Mathematics and Systems Analysis} \\
\textit{Aalto University School of Science}\\
Otakaari 1, 02150 Espoo,  Finland\\
lasse.leskela@aalto.fi}\\
\and
\IEEEauthorblockN{ 2\textsuperscript{nd} Ville Kotovirta}
\IEEEauthorblockA{\textit{ Collective and collaborative AI  } \\
\textit{VTT Technical Research Centre of Finland}\\
P.O. Box 1000, FI-02044 VTT, Finland \\
ville.kotovirta@vtt.fi}\\

\IEEEauthorblockN{4\textsuperscript{th} Hannu Rummukainen}
\IEEEauthorblockA{\textit{Collective and collaborative AI } \\
\textit{VTT Technical Research Centre of Finland}\\
P.O. Box 1000, FI-02044 VTT, Finland \\
hannu.rummukainen@vtt.fi}\\
\IEEEauthorblockN{5\textsuperscript{th} Tomi R\"aty}
\IEEEauthorblockA{\textit{ Collective and collaborative AI } \\
\textit{VTT Technical Research Centre of Finland}\\
P.O. Box 1000, FI-02044 VTT, Finland\\
tomi.raty@vtt.fi}
}

\maketitle

 \begin{abstract}
        The use of quantum computing in graph community detection and regularity checking related to Szemeredi's  Regularity Lemma (SRL) are demonstrated with D-Wave Systems' quantum annealer and simulations. We demonstrate the capability of quantum computing in solving hard problems relevant to big data. A new community detection algorithm based on SRL is also introduced and tested. In worst case scenario of regularity check we use Grover's algorithm and quantum phase estimation algorithm, in order to speed-up computations using a quantum gate computers. 
    \end{abstract}   

\section{Introduction}

We are entering the exciting era of quantum computing. There is hope that this new computing paradigm is also useful in studying hard problems in the analysis of large graphs emerging from big data.

The use of quantum computing needs a new mindset. Probably the simplest avenue in this direction is the so-called quantum adiabatic computing and quantum annealing in particular which can be used in almost any optimization task. Quantum gate computing could be used for many more problems than a quantum annealer, but there each algorithm is an untrivial milestone in itself like the celebrated Shor's algorithm. 

The quantum annealing  hardware is reaching over $5000$ qubits (qubits are quantum objects that replace bits in ordinary computation) in the near future while gate computers are developing at the somehow more modest pace. The D-Wave Systems company has made quantum annealing available as a cloud service allowing experiments with over $2000$ qubits as well as an easy to use interface to the system. Hybrid classical-quantum algorithms, available also for D-Wave machines, make solving larger problems possible.

In this work, we consider use of quantum annealing for graph partitioning and, in particular, graph community detection with a new algorithm. We hope that our work will motivate other similar studies in the big data area. Our aim is to demonstrate the potential of quantum computing in analysing large graphs. Such an approach is likely to push the boundary of graph sizes in which good quality solutions can be found. We also demonstrate how quantum annealers are used in concrete cases.       

A starting point of our work is Szemer\'edi's Regularity Lemma (SRL), a cornerstone of extremal graph theory, see e.g. \cite{fox2017}. SRL justifies a kind of stochastic block model structure of bounded complexity for all large graphs. SRL has had a great impact in the theoretical study of large graphs and that is why it can have a decisive role in future big data analysis as well.

SRL's  key concept is an $\epsilon$-regular bipartite graph. It is a bipartite graph in which link density deviations in any sub-graphs are bounded by some positive $\epsilon$. This means that such a bipartite graph is close to random one.  In SRL, the $\epsilon$-parameter can be chosen to be arbitrarily small. SRL states, roughly speaking, that any large graph has a partitioning of nodes to a bounded number of sets in which links between parts follow the $\epsilon$-regularity.  

Regular partitioning can be found in polynomial time. However, deciding $\epsilon$-regularity of a bipartite graph is co-NP-complete problem. We show that the regularity check is a binary quadratic optimization problem. It has the form that can be solved with the D-Wave quantum annealer. As a result, such optimization can be a very hard problem that can be of interest to test efficiency of quantum annealers and we use it as an example of a hard problem arising in large graph analysis. 

It appears that the same optimization task, as used in regularity check, can be used to find communities of an arbitrary graph. This novel algorithm does not need any parameters besides the adjacency matrix. The stochastic block model of communities can be seen as a particular realization of regular partition. In future we shall study a more general case of SRL from this point of view.  The suggested algorithm has some advantages over implementing the standard community detection algorithm on D-Wave \cite{negre}. Namely, it requires only $1$ qubit per graph node and no prior knowledge of the number of communities. In standard approach, each node requires $k$ times more qubits, in which $k$  is the maximal number of communities. Since qubits are scarce resource, this difference is significant.

We anticipate that quantum annealing can produce better quality solutions for large graph problems than classical computation. Interestingly, in \cite{negre} evidence pointing to this direction was already found. The quantum community detection algorithm found the best quality solution, measured in so-called modularity metrics, compared to any previous method. This was the case of well-known test graph with only $34$ nodes, so-called Zachary Karate Club graph. Another point could be that such good solutions can be found in larger scales than is possible with classical computing.    

We test our ideas using D-Wave System and simulations. We also consider the performance of our community detection algorithm using stochastic block models and discuss further challenges.

\section{Regularity check as an optimization problem }

Let $G=G(A,B,d(A,B))$ denote a bipartite graph, in which the set of nodes of $V$, is divided into two disjoint sets $A$ and $B$, see Fig. 1. The number of nodes (cardinality) of a set of nodes $X$, is denoted as $\abs{X}$. The number of links connecting two arbitrary subsets of $V$, $X$ and $Y$ is denoted as $e(X,Y)$. Similarly the link density between two disjoint node sets $X$ and $Y$ is by definition:
$$
d(X,Y)=\frac{e(X,Y)}{\abs{ X}\abs{Y}}
$$
The binary adjacency matrix of $G$ is denoted as $A$. Value $(A)_{i,j}$ is one only if there is a link between nodes $i$ and $j$.
\begin{figure}
\centering
\includegraphics[width=8cm]{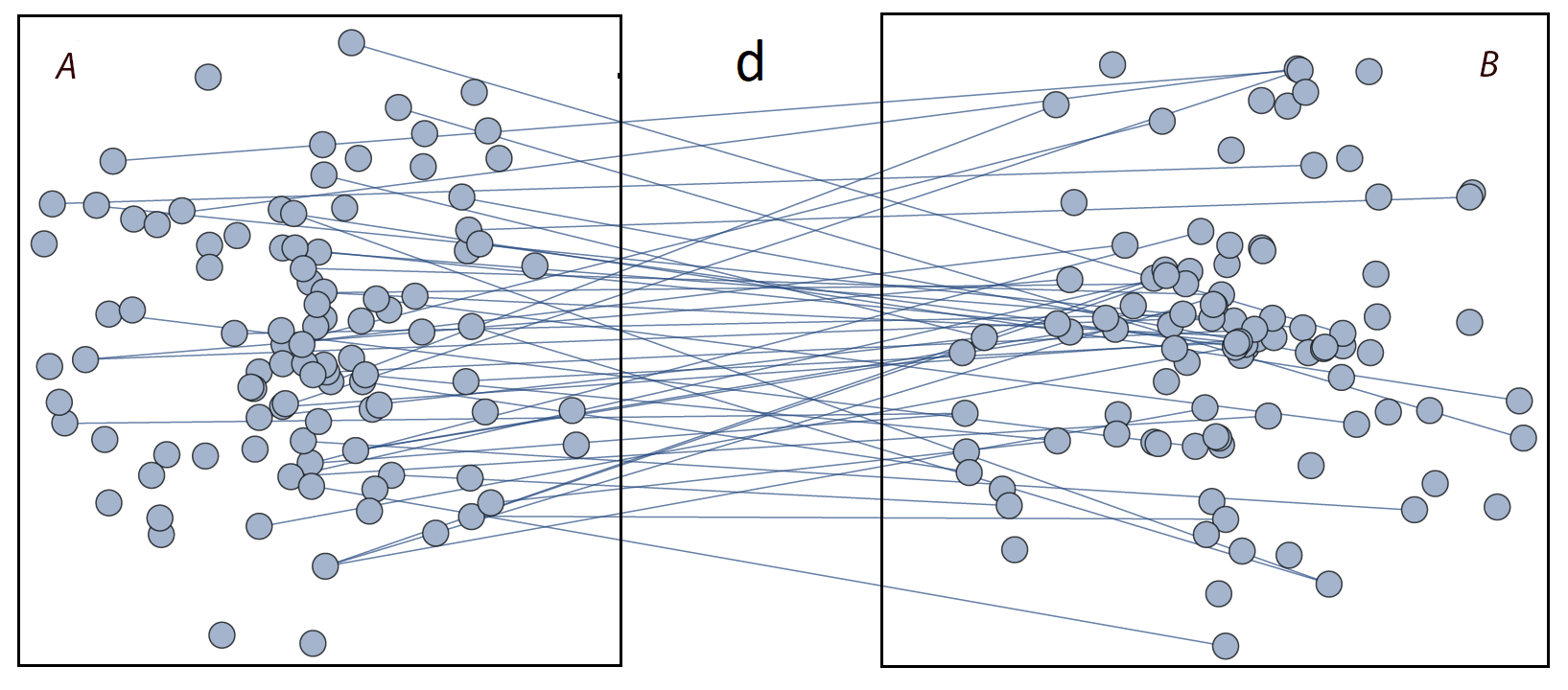}
\caption{A bipartite graph, where links can exist only between nodes in sets $A$ and $B$; link density $d$ is fraction of pairs $(i,j)\in A \times B$ that have links. }
\label{fig:universe}
\end{figure}

A bipartite graph $G(A,B,d(A,B))$ is called $\epsilon$-regular if, for all subsets $X\subset A, Y \subset B$, the following is true:
$$
\abs{ X}\abs{Y}(d(A,B) - d(X,Y)) = O(\epsilon \abs{A}\abs{B}).
$$
This definition was given by T.Tao \cite{tao}. Regularity is a key concept in Szemer\'{e}di's Regularity Lemma \cite{Szeme76}. In the standard definition, it is required that link density deviation is $\epsilon$-small for all large subsets. Tao's definition suits us better since there is not such constraint on subset sizes. Small sets are regular in this sense, just because of their small sizes assuming large sets $A$ and $B$. 

For each bipartite graph, define a function: 
$$
L:2^{A}\times 2^{B}\rightarrow \mathbb{Q},\quad L(X,Y) = \abs{ X}\abs{Y}(d(A,B)-d(X,Y)),
$$
in which $2^M$, denotes set of all subsets of $M$ and $\mathbb{Q}$ is set of rational numbers. For better interpretation, we rewrite:
\begin{eqnarray*}
L(X,Y) = \abs{ X}\abs{Y}d(A,B)-e(X,Y) =\\= \mathbb{E}_d e(X,Y)-e(X,Y),
\end{eqnarray*}
in which $\mathbb{E}_d$ denotes the expectation operator in a random bipartite graph with the link probability equal to $d=d(A,B)$. As a result, $L(X,Y)$ is simply deviation of the number of links in the subgraph, induced by $X \cup Y$, from the expected number of links in the random bipartite graph. In this work we work with minimization of $L$. Maximization is done similarly using $-L$ as the cost function. As a result, $\min L$, corresponds finding the largest fluctuation that exceeds most the expected value $\mathbb{E}e(\cdot,\cdot)$. 

Quantum annealers, like D-Wave,  are capable of solving  quadratic binary optimization problems (qubo):
\begin{eqnarray}
\label{Qubo}
\min_{s} \sum_{i,j}(J_{i,j}s_is_j+h_is_i),
\end{eqnarray}
in which $J$ and $h$ are fixed matrix-valued parameters and  $s$ is a vector of binary variables.

We can easily write the minimization of $L$ in this form. 
For given subsets $X$ and $Y$ assign the values of binary variables $s_i\in\{0,1\}$ to all nodes in $i\in V$:
$$
   i\notin X\cup Y  \Rightarrow s_i=0, \quad i\in X \cup Y \Rightarrow s_i=1. 
$$
As a result:
$$
\abs{X}= \sum_{i\in A} s_i, \quad \abs{Y}= \sum_{i\in B}s_i.
$$
Similarly:
$$
e(X,Y)= \sum_{i\in A, j\in B} a_{i,j} s_is_j,
$$
in which $(A)_{i,j}=a_{i,j}$ is the adjacency matrix of $G$.

Using these notations and because by definition
$$
\abs{ X} \abs{Y} d(X,Y)=e(X,Y),
$$ 
we can write the above program as a qubo: 
\begin{eqnarray}
\label{qubo}
(X_1^*,X_2^*)= \argmin_{X,Y}\sum_{i\in V_1,j\in V_2}(d(A,B) - a_{i,j}) s_is_j,\\\nonumber
X=\{i\in A: s_i=1\}, Y=\{j \in B: s_j=1\})
\end{eqnarray}
Going trough all the configurations of $s$-variables is equivalent of going trough all subsets $X$ and $Y$.

Define the following block matrix $M$:
\begin{eqnarray*}
\{i,j\}\subset V_1 \quad or\quad \{i,j\}\subset V_2 \implies (M)_{i,j}=0
\end{eqnarray*}
and otherwise
$$
(M)_{i,j}=d(A,B)-a_{i,j}.
$$
Using $M$ we can write: 
\begin{eqnarray}
\label{qubo2}
L(X,Y) = \frac{1}{2}\sum_{i,j}s_iM_{i,j}s_j= \frac{1}{2}(s,Ms),
\end{eqnarray}
in which $s$ is the vector of $s$-variables, $(\cdot,\cdot)$ is inner product of vectors and the summing is over all indices $i$ and $j$.

$\epsilon$-regularity of a bipartite graph means that $L$ is $\epsilon $-bounded function for that graph. As a result, finding global minimum and maximum of this function would resolve the $\epsilon$-regularity check decision problem. Since this problem is co-NP-complete, it is likely that there are no efficient algorithms for finding the minimum and maximum of $L$ function for all graphs. For this reason, the optimization of $L$ can provide a needed challenge for quantum computing to demonstrate its power.

\section{Community detection algorithm}
\subsection{Stochastic block model}

As we see in the following, finding the $\min L$ in a bipartite graph can be seen also as a basic operation in finding communities in a graph. We consider the case when the graph has communities generated from a stochastic block model (SBM), for a review see \cite{abbe}. 

SBM($n,k,P,D$) is a generative probabilistic graph model defined as follows. Here $P$ is a probability distribution on $[k]=\{1,\dots,k\}$ and $D$ is a symmetric $k$-by-$k$ matrix with entries $D_{i,j} \in [0,1]$. The model is generated by first sampling node labels $\sigma(1),\dots,\sigma(n)$ independently from $P$, and then creating a random graph on node the set $V$ by linking each unordered node pair $\{u,v\}$ with probability $D_{\sigma(u),\sigma(v)}$, independently of other node pairs. The node labeling $\sigma: V \to [k]$ partitions the node set into $k$ disjoint communities $V_i = \sigma^{-1}(i)$, so that
$$
V = V_1 \cup \cdots V_k.
$$
Conditionally on the node labeling $\sigma$, the nodes between communities $V_i$ and $V_j$ are hence linked with probability $D_{i,j}$. The resulting random graph is denoted as  $\mathcal{G}(n,k,P,D)$.

\subsection{Community detection algorithm}

In our previous works, we have extensively referred to SRL as a basis for graph analysis using various SBMs as a modeling space \cite{reittuetall,nepusz2008,pehkonenreittu,reittujoensuu,reittubazsonorros,reittuBigData2018}. Here we introduce another contact point between SBM and SRL.

Assume that a graph $G$ is drawn from $\mathcal{G}(n,k,P,D)$ as described above. We also assume that $n$ is large enough. 

The first step is to find a bipartite subgraph, $G'$ of $G$:
\begin{itemize}
    \item divide nodes of $G$ into two disjoint sets $A$ and $B$, by tossing a fair coin for each node
    \item $G'$ inherits all links from $G$ that join $A$ and $B$ while all links inside $A$ and $B$ are deleted.
\end{itemize}
This procedure is schematically shown in Fig. \ref{split1}

\begin{figure}
\centering
\includegraphics[width=6cm, height= 6cm]{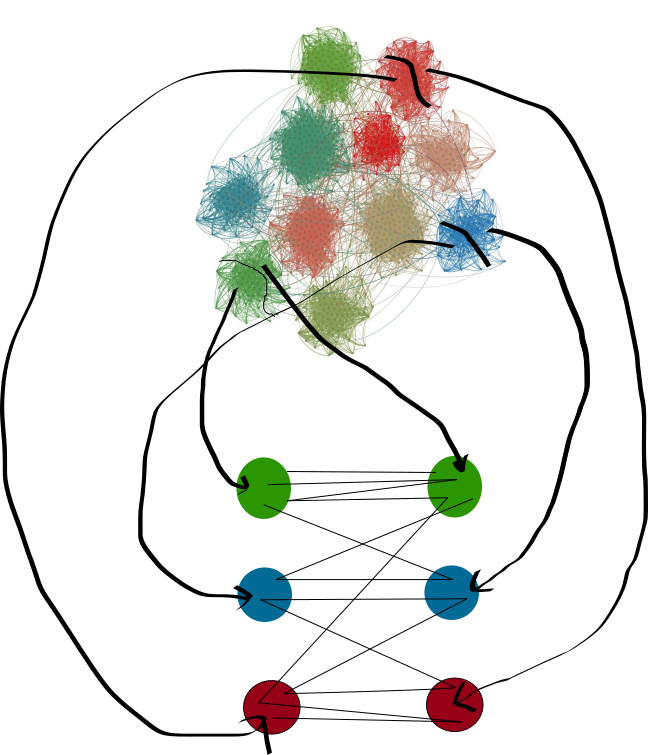}
\caption{Generation of a bipartite graph $G'$ by a random split. At the top is a generic graph with communities indicated by different colors. Nodes are divided into two random sets, say, by tossing a fair coin. Each community is roughly split into two parts, one at the left ($A$) and other in the right ($B$). For three communities (indicated by green, blue and red balls) the bipartite graph $G'$ is shown in the lower part. Only those links in $G$ joining $A$ and $B$ are preserved in $G'$.   }
\label{split1}
\end{figure}

We denote: $A_i := A \cap V_i$ and $B_i := B \cap V_i$ with sizes $a_i:=\abs{A_i}$ and $b_i:=\abs{B_i}$ for $i=1,\cdots, k$. It is clear that random variables $(a_1,\cdots,a_k,b_1,\cdots, b_k)$ have a multinomial distribution with expectations $\mathbb{E}a_i = \mathbb{E} b_i = n P_i/2$ for all $i$. For large $n$ these random variables are well concentrated around their expected values. 

Denote by $d$ the link density of bipartite graph: $d = e(A,B)/(\abs{A}\abs{B})$. For a large graph, $d$ is close to the expected link density of the original graph $G$, $d(G)$, with high probability. We require the following
\begin{equation}\label{condition}
  d(G) - D_{i,i} < 0, \forall i.
\end{equation}

This inequality means that all communities have internal density above the average density. For large graphs, we also have with high probability: 
$$
d - D_{i,i}<0, \forall i. 
$$

It is required that the Condition \ref{condition} holds when $G$ is replaced by a subgraph of $G$ in which arbitrary communities are deleted.

We do not provide a lengthy proof of the last claim. Typically probabilistic estimates are exponential, so this statement has high probability already with moderate graph sizes. 

The idea behind community detection is the following. If there are bigger densities of links inside the communities than those between the communities, then the communities in the split graph are associated with the denser parts of the corresponding bipartite graph. As a result, there is a chance that communities can be found with the help of $\argmin L$ applied to the split graph.
\begin{figure}
\centering
\includegraphics[width=4cm, height=6.5cm]{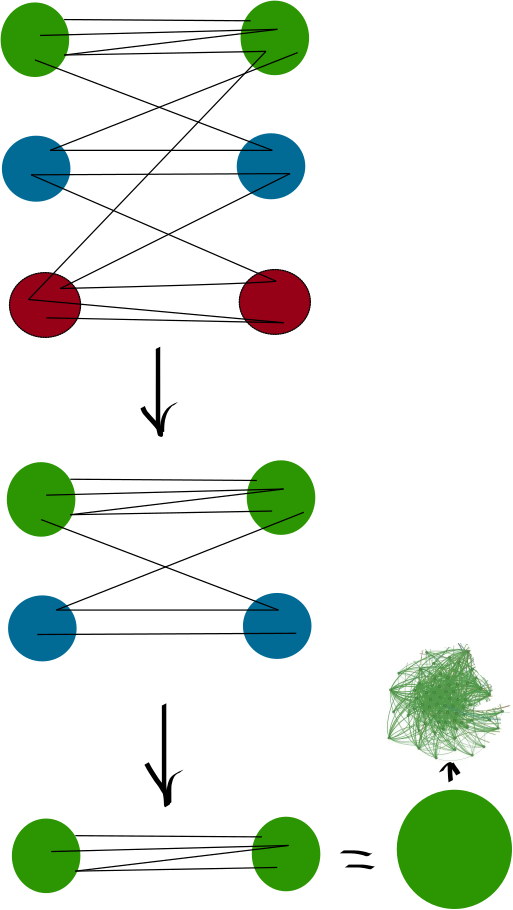}
\caption{The stages of community detection algorithm on a bipartite graph with split communities indicated by colors. At each stage, some communities are deleted. The program ends when only one community is left, the green ball in the figure. By deleting the found community from the whole graph, one could proceed to find the next community and so on until all are found.   }
\label{split}
\end{figure}

The next step of the algorithm is to construct the $L$-function for the bipartite graph of the split communities described above. The output of the algorithm is the subgraph induced by  $\argmin L$.  The algorithm works correctly if the following conjecture is true:
\begin{conjecture}\label{conjecture}
Consider graph $G$ that is generated from $\mathcal{G}(n,k,P,D)$ with communities $V_1,\cdots,V_k$ and Condition \ref{condition} holds. Construct a random evenly split bipartite graph $G'$ with bipartition $(A,B)$, $A = \cup_{i} A_i$ and $B = \cup_i B_i$ and in which sets with indices $i$ are subsets of community $V_i$ for $i=1,\cdots,k$. Let $L(\cdot,\cdot)$ correspond to graph $G'$. Then with probability tending to $1$ as $1-\exp(-n^z)$ when $n\rightarrow \infty$ and with some constant $z>0$, the following holds:
$(X^*,Y^*)=\argmin L(X,Y)$ $\Rightarrow$ there is a proper subset of indices $I \subset \{1,\cdots,k\}$ such that $X^* = \cup_{i\in I} A_i$ and $Y^* = \cup_{i \in I} B_i$.
\end{conjecture}

We do not possess a full proof of this claim. As a first sketch, we consider optimization of expected $L$-function conditional to the sizes of split communities. The basic setting is the same as in Conjecture \ref{conjecture}. We denote $x_i = \abs{X \cap A_i}$ and $y_i = \abs{Y \cap B_i}$. These integers are bounded by $0 \leq x_i \leq a_i$ and $0 \leq y_i \leq b_i$. In shorthand, we write $x \in [0,a]$ and $y \in [0,b]$ and recall that such constraints are usually referred as box constraints \cite{floudas}. In these notations, we have:
$$
 L(X,Y)
 = \sum_{i,j} (x_i y_j d(G') - e(X \cap A_i, Y \cap B_j)).
$$

Let us condition with respect to $a_1, \cdots, a_k, b_1,\cdots b_k$ and take the expectation of $L$ over the SBM:
$$
 L_1(X,Y)
 := \mathbb{E}L(X,Y)
 = \sum_{i,j} x_i x_j(d - D_{i,j}),
$$
in which $d$ is the expected link density of the bipartite graph conditionally on the underlying community structure. We assume that we are in the high-probability event when the Condition \ref{condition} holds.  

\begin{proposition}\label{proposition}
$(X^*,Y^*) = \argmin L_1(X,Y)$ has the same structure as in Conjecture \ref{conjecture} in the sense that for the $(X^*,Y^*)$, $(x_i,y_i) \in \{(0,0), (a_i,b_i)\}$ for all $i$, and there exists index $i$ such that $(x_i,y_i)=(0,0)$. 
\end{proposition}

\begin{proof} (A sketch)
Let us consider a relaxation of the optimization problem in which the integer variables $x_i$ and $y_i$ are replaced by their continuous counterparts, keeping the box constraints. We use the same symbols $x_i$ and $y_i$ to correspond to the global minimum of $L_1$ within the box constraints and in the continuous variables. Denote 
$$
 L_1(X^*,Y^*) = \sum_{i,j} x_i m_{i,j} y_j.
$$
There is no global optimum strictly inside the box. To see why, denote the partial derivatives of $L_1$ by $f_i := \partial_{x_i}L_1(X^*,Y^*)$ and $g_i: = \partial_{y_i} L_1(X^*,Y^*)$. A global optimum inside the box would lead to $$L_1(X^*,Y^*) = \sum_i f_i x_i=0,$$ which is not the global minimum since $L_1$ can take negative values, say, when we take just one community, $L_1(A_1,B_1) = a_1(d -  D_{1,1})b_1<0 $ by Condition \ref{condition}. That is why the optimal point is on the boundary of the box. It must also be in the 'corners' of the box, meaning that components have values $0$ or have the maximal possible value. In a point that is on the box boundary but not at a corner point, the gradient of $L_1$ is pointing inside the box volume and by moving towards some direction, provided the gradient is not perpendicular to the boundary, one could reduce the value of $L_1$, this is impossible only if the point is in one of the corners of the box. If the gradient is perpendicular to boundary of the box, we would have $(\nabla_x L_1)_i= c \delta_{\alpha,i}$ and $(\nabla_y L_1)_i= c' \delta_{\beta,i}$ and $L_1= (x,\nabla_x L_1)=x_\alpha c \sum_i(d-D_{\alpha,i})y_i= y_\beta c'\sum_i(d-D_{\beta,i})x_i$ in that point. As a result $c=c'$ and $\alpha = \beta$ and $L_1=x_\alpha y_\beta(d-D_{\alpha,\alpha})\geq n_\alpha m_\alpha(d-D_{\alpha,\alpha})$. The found lower bound corresponds to a choice of a corner point and thus such a solution has lower energy than the suggested orthogonal to the boundary of the box. It is also easy to see that if $x_i=a_i$, then also $y_i=b_i$. This is due to $a_i \approx b_i$ and Condition \ref{condition}. As a result, the global optimum is just one of the corner points. The corner points of the box have integer coordinates and as a result the found minimum is a solution of the original integer problem. 
\end{proof}

To prove Conjecture \ref{conjecture}, we need some probability concentration inequalities like Chernoff bounds for known distributions with which we are dealing. The martingale argument may be used as was shown in an analogous case \cite{reittuData}. Proposition \ref{proposition} shows that the Conjecture holds on average and provides a starting point for the proof. Then one should use concentration inequalities of probability theory to show that the solution of stochastic problem is the same with high probability.   

After the first step, the algorithm proceeds similarly on the found subgraph corresponding to  $\argmin L$. It runs until only one community is left. Then the found community is deleted from $G$ and the whole process is repeated until all communities are found, see Fig \ref{split}.

\section{Simulations and experiments with D-Wave machine}
\subsection{Brief description of the D-Wave Leap system }

D-Wave Systems Inc. has published a quantum computing cloud service, Leap \cite{leap}, for free trial and an option for buying quantum processing time. The Leap provides immediate access to a D-Wave 2000Q quantum computer or annealer. The computer has up to $2048$ qubits and the service  provides support for users such as the demos, interactive learning material and the Ocean software development kit (SDK) with suite of open-source Python tools and templates.

When one has a qubo (\ref{qubo}) in the form of (\ref{qubo2}), it is quite straightforward to implement and run it on the quantum computer. The size of the problem is restricted by the connectivity between qubits in the D-wave 2000Q Chimera architecture and the number of qubits available. 

For an arbitrary problem (\ref{qubo}) an embedding is needed and that can drastically reduce the size of the problem that can be solved. For large problems, a hybrid approach is suggested, in which the problem is split into smaller pieces and part of the computations are done classically, so-called qbsolver. The Ocean SDK provides functions for automatic embedding as well as for solving a large qubo by qbsolver. D-Wave has announced that the next generation version of the quantum computer with over $5000$ qubits and added connectivity between qubits would be available in mid-2020.  

\subsection{Experiments}

\subsubsection{Regularity check of a cortical area graph}\label{1}
Our first example is a small bipartite graph in Fig. 4, with $18$ nodes taken from \cite{nepusz2008}, in which SRL was used to analyse connections between cortical areas in a brain of a primate. This graph represents one regular pair of a bigger graph. 
 \begin{figure}
 
\includegraphics[scale=0.3
]{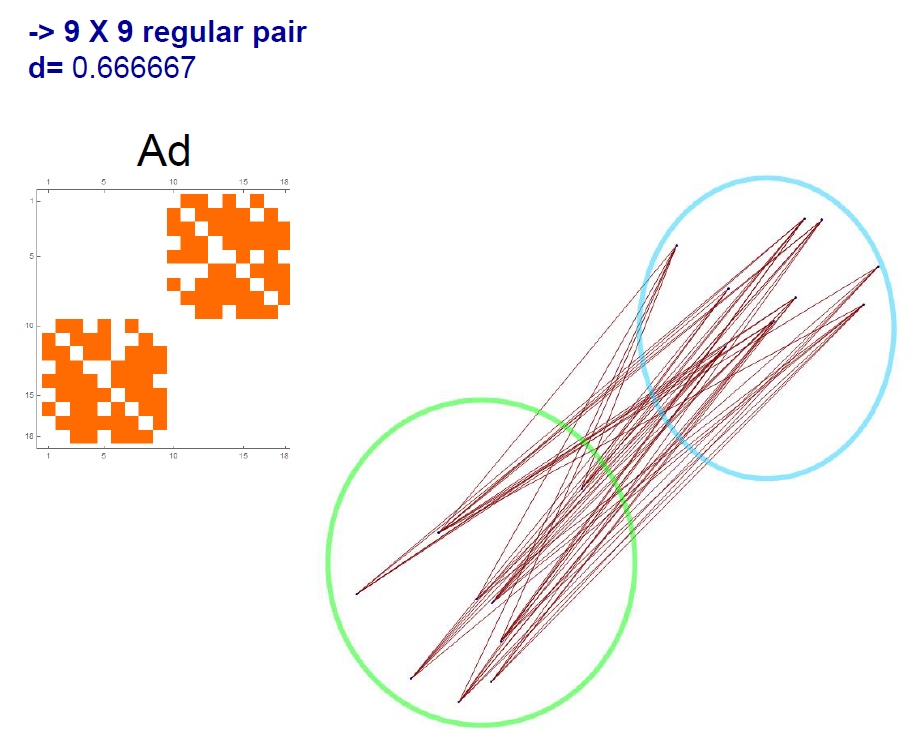}
\caption{Adjacency matrix (Ad) and the bipartite graph}
\end{figure} 

Our task is to find a solution of the qubo (\ref{qubo}) for this graph which is equivalent to the regularity check. In this case, the number of possible subset pairs is just $2^{18}=262144$. As a result, the full search is possible. D-Wave finds the solution in a default time (microsecond). 

{\it The Result:} D-Wave finds the exact solution in one run taking microsecond. 

\subsubsection{Regularity check of random bipartite graphs} Next we solved qubo (\ref{qubo}) for a random bipartite graph. In first experiments both segments have $50$ nodes. The links between node pairs were drawn independently at random with a fixed probability $=0.2$. 

In this case, it would be expensive to find the global minimum by exhaustive search. However, using the mixed-integer linear programming (MILP) solver CPLEX 12.9 with the model described in \cite{dash}, we established that the found solution is the global minimum.\\  Interestingly, McGeoch and Wang \cite{mcgeochWang} compared the CPLEX and the D-Wave machine Vesuvius 5 with $439$ qubits on randomly generated Ising qubo instances. The Ising model was generated on a Chimera subgraph, so the  $J$-matrix in qubo (\ref{Qubo}) was a weighted adjacency matrix of the Chimera subgraph. Such a problem is a sparse one. Dash \cite{dash} noted that with a suitable MILP model, CPLEX found the solution to the McGeoch and Wang instances very quickly and in a comparable time with D-Wave: For example with $512$ nodes, the average time was $0.19$s. 

D-Wave was used in hybrid quantum-classical mode when a part of the problem was solved on an ordinary computer and only smaller sub-problems are solved on D-Wave, so called qbsolver algorithm. This allows treating optimization problems with much larger number of variables than the number of qbits available in D-wave. 

It appears that our case of a dense bipartite graph is much harder than the problems that McGeoch, Wang and Dash studied. In the described $50\times 50$ node graph the required time to find the optimal solution was $4.5$ hours and to verify that the solution was the global minimum, took an additional $16$ hours. The computer had four 2.7 GHz Xeon E5-4650 CPUs, with a total of $32$ cores and $512$ GB RAM.
 
{\it The Result:} Both D-Wave and simulated annealing algorithm produce the same least energy solution when around $1000$ instances were examined. This solution is the global minimum found with CPLEX. The required time at D-wave is again very small, less than a second, if the queuing time to the Leap cloud service is neglected. Simulated annealing needed around one minute. 

The sample graph had a density around $0.2044$ and the found subgraph that is the solution of qubo (\ref{Qubo}) had density around $0.313$.  The corresponding  non-zero blocks of adjacency matrices are shown in Figures 5-6. 

\begin{figure}
 \label{full}
 \begin{center}
  \begin{minipage}[b]{0.4\textwidth}
 
    \includegraphics[width=\textwidth]{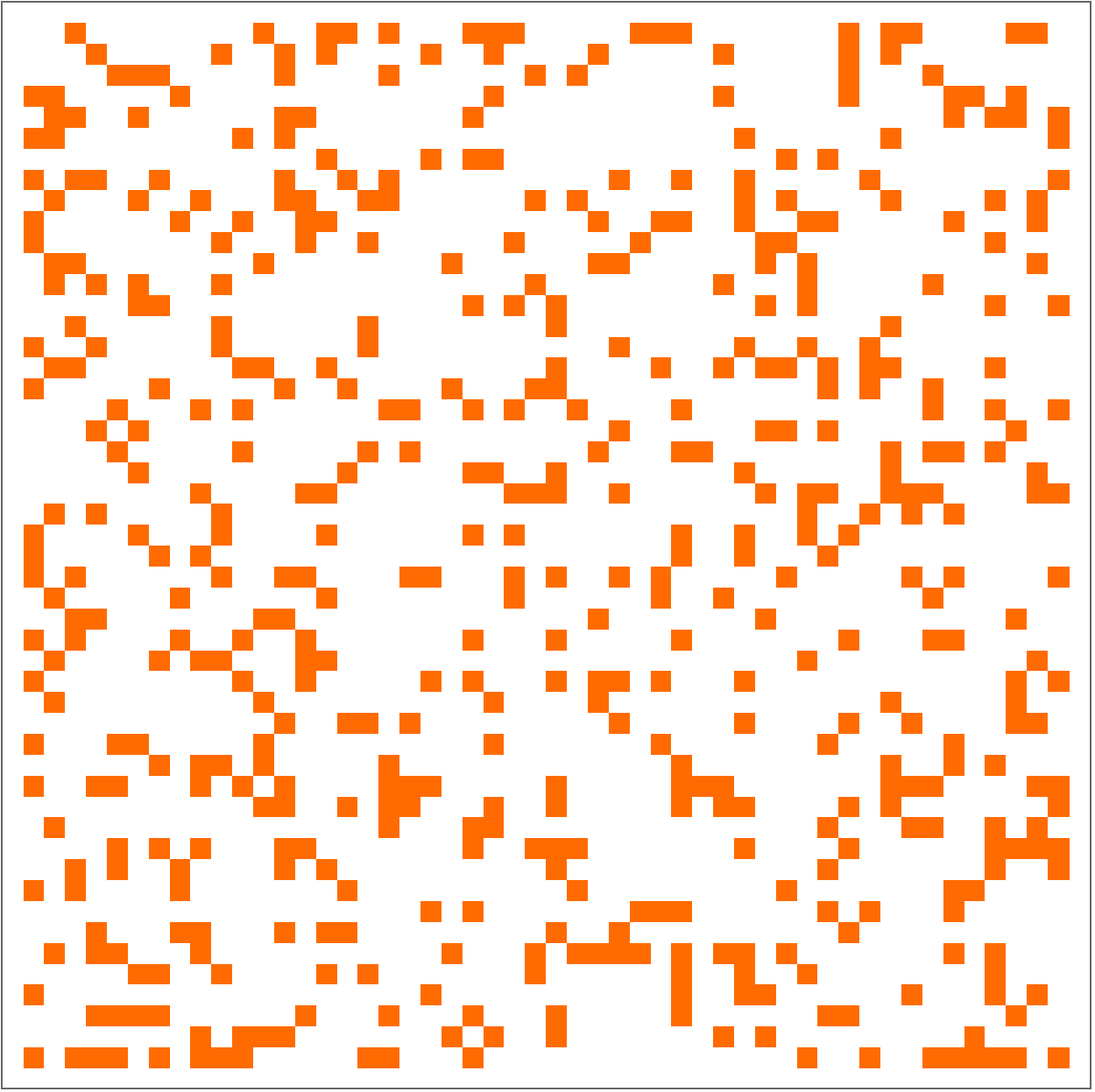}
    
    \caption{Adjacency matrix of a bipartite graph}
     \vspace{0.5cm}

  \end{minipage}
  \hfill
  
  \begin{minipage}[b]{0.2\textwidth}
  
    \includegraphics[width=\textwidth]{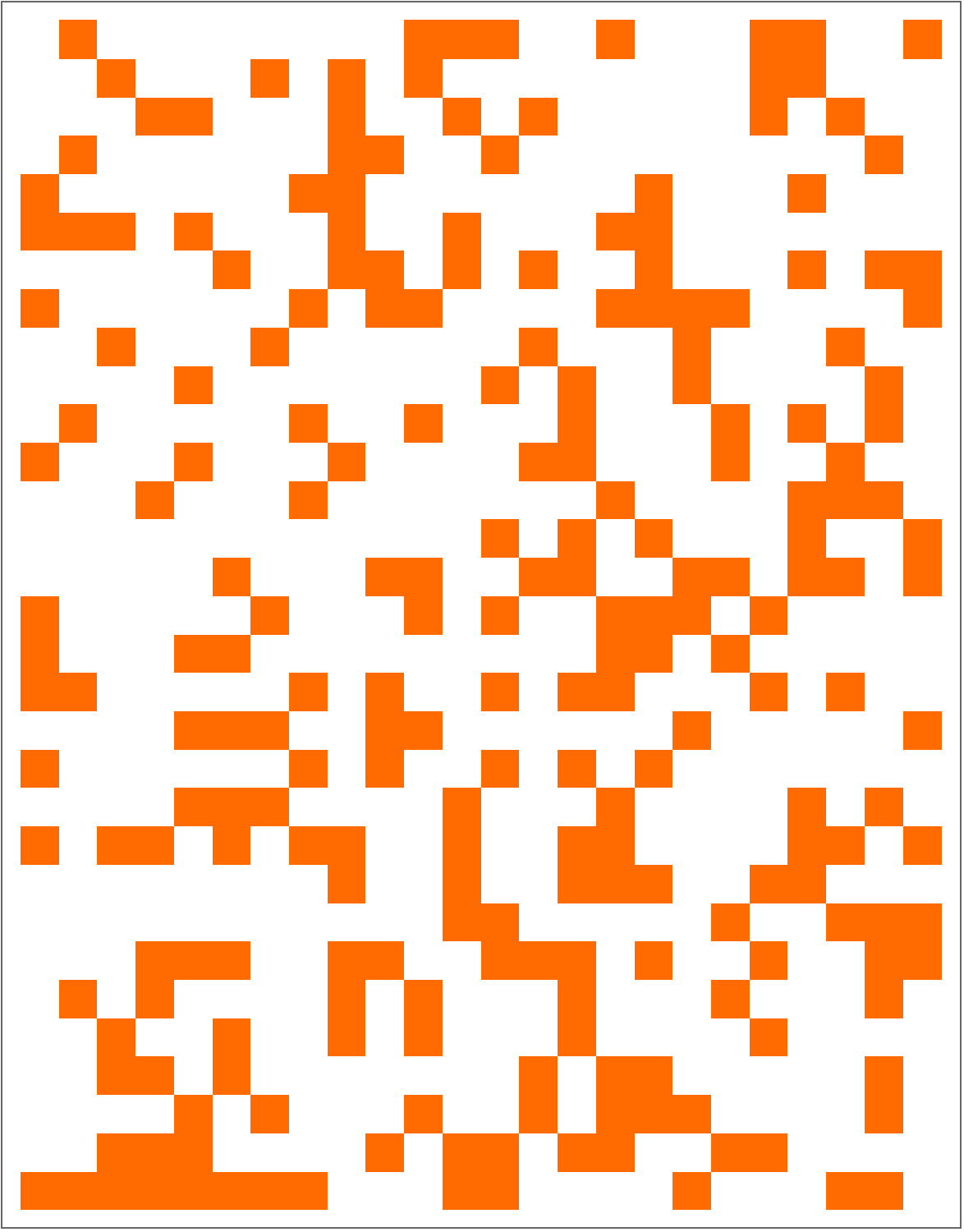}
    
    \caption{Adjacency matrix of largest irregularity pair of size $31\times 24$ of the bipartite graph in Fig. 5}
  \end{minipage}
  \end{center}
\end{figure}

\subsubsection{Execution times on bipartite graphs}
 We tested solving the qubo (\ref{Qubo}), using large random bipartite graphs. For instance, in the case of $400$ nodes it is not possible to use D-Wave in a similar way as in the case of $100$ nodes. Instead, we used the simulated annealing algorithm, qbsolve with D-Wave or simulated annealing provided by the Leap system. The last two methods means that the problem is split into smaller pieces and the pieces are solved by classical or quantum annealing. In this way lager problems can be solved on D-Wave machine.  
 
 For such scales of random dense graphs with several hundreds of nodes CPLEX becomes also impractical, due to long execution time. Simulated annealing also slows down. For $2000$ nodes, the time to find approximate solution is several hours on a laptop. In this case, it is not possible to verify with CPLEX whether a global minimum was found. As a results, such execution times are only lower bounds of the optimization time.
 
  \begin{figure}
  \begin{center}
 \label{fulop}
\includegraphics[scale=0.4
]{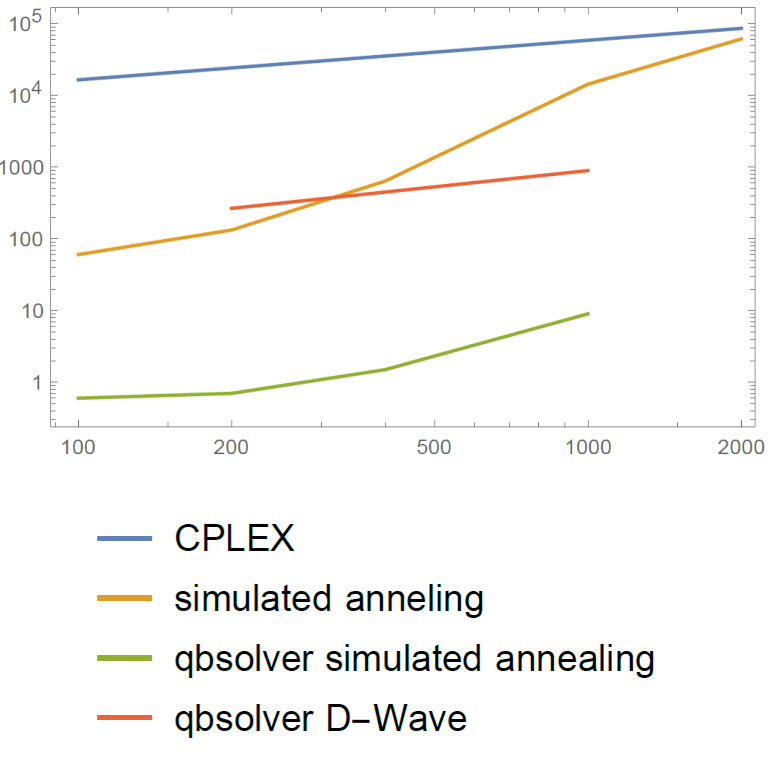}
\caption{Execution time in seconds of different optimization methods as a function of graph size in log-log scale. CPLEX can find exact solutions for small graphs and a poor solution for large graphs. In case of qbsolver with D-Wave annealer, the queuing time for cloud service is not filtered away. }
\end{center}
\end{figure} 

The result is shown in Fig. 7. The heuristical qbsolver classical annealer is the quickest, however producing slightly lower quality solutions for large graphs than the usual simulated annealer. The D-Wave qbsolver shows almost a constant time and, in this case, the queuing time is not filtered away. This may suggest that for extremely large cases, D-wave-assisted qbsolver is the winner in terms of time and quality. CPLEX can find exact solutions for small scales, but it is very slow and produces poor quality solutions for large graphs.  

We hypothesize that the qubo (\ref{Qubo}) for regularity check of a random bipartite graph can be a hard problem to solve already for moderate sizes of the underlying graph and can be used for testing quantum annealers. On the other hand, it suggests even some simple optimization problems emerging from large data can be only solved exactly with future quantum computers. 
\subsubsection{Small scale community detection}

We tested our community detection algorithm using D-Wave and simulated annealing for a bipartite graph with $100$ nodes and two communities. The adjacency matrix and the graph are shown in Figs. 8-9

\begin{figure}
\begin{center}

\includegraphics[scale=0.4]{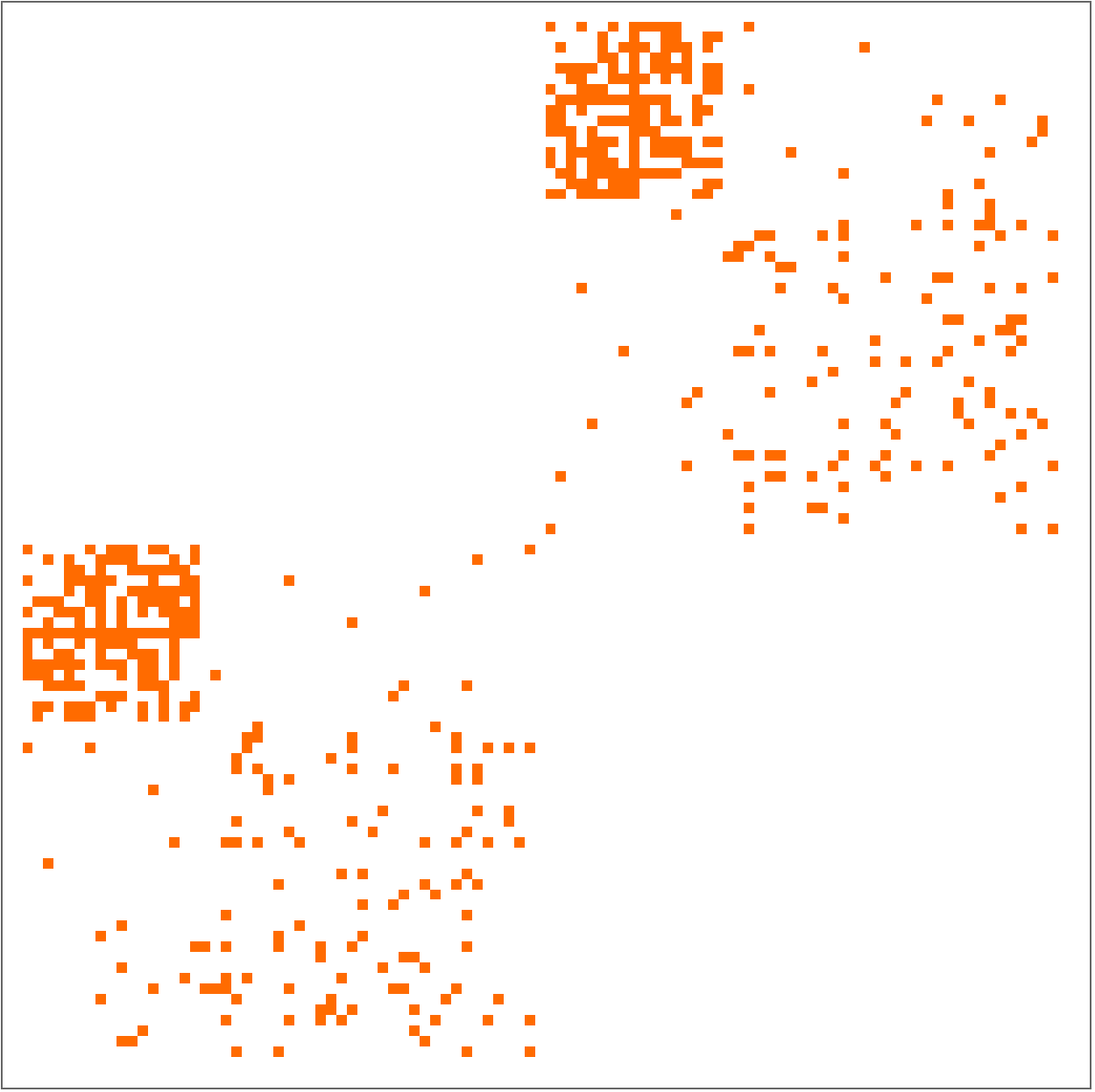}
\caption{Adjacency matrix of a $100$ node bipartite graph with two communities, seen as blocks of higher concentration of $1$s}
\end{center}
\end{figure} 

\begin{figure}
\begin{center}
\includegraphics[scale=0.17,angle=90]{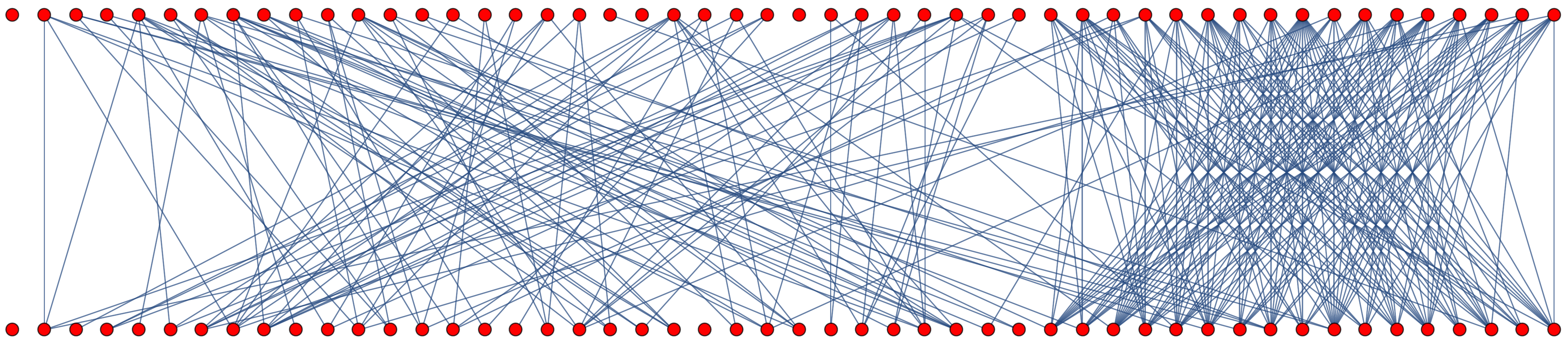}
\caption{Bipartite graph with two communities, the denser community is at the right end of the plot}
\end{center}
\end{figure}

{\it The Result:} both quantum - and classical annealers found the correct communities with 100 percent accuracy. 

Next experiment was done using only classical annealer because the graph had $200$ nodes which cannot be embedded directly in the Chimera graph of D-Wave. The graph has three communities, see Fig. 10. In the first round, the largest and sparsest community was dropped out. At the second round, one of the communities was left alone. {\it As a result} the algorithm found all three communities perfectly. 
 \begin{figure}
\begin{center}
\includegraphics[scale=0.2
]{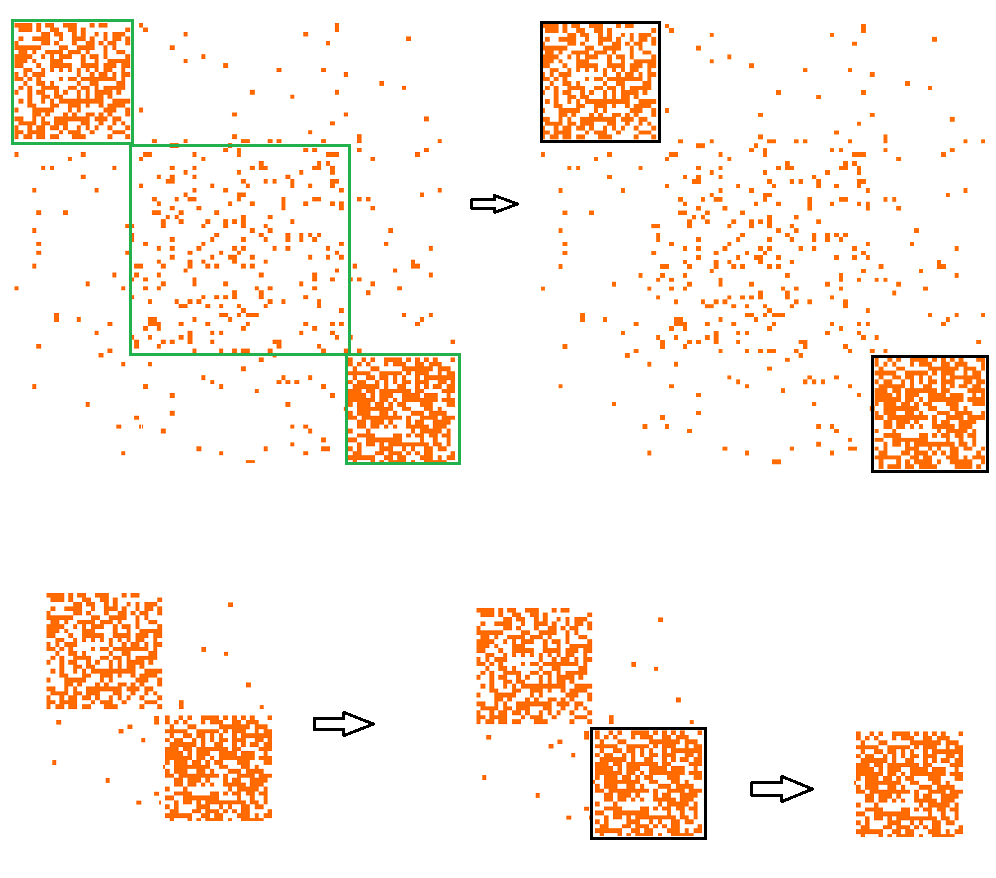}
\caption{From upper left corner: $3$ communities with internal links inside green boxes. The first round of iterations selects two most dense communities shown in black boxes. The next round picks one of the communities as an output. Result: perfect detection of all $3$ communities}
\end{center}
\end{figure} 
\subsubsection{Towards large scale community detection}

We consider a larger case of graph that has $2000$ nodes and $10$ communities. The adjacency matrix of the split bipartite graph is shown at the top of Fig. 11. The diagonal blocks of the two non-zero large blocks corresponds to links inside the communities. The darker color indicate higher density of links. In this case, the algorithm works as stated in Conjecture \ref{conjecture}, the output is one community.    

\begin{figure}
\begin{center}
\includegraphics[scale=0.4
]{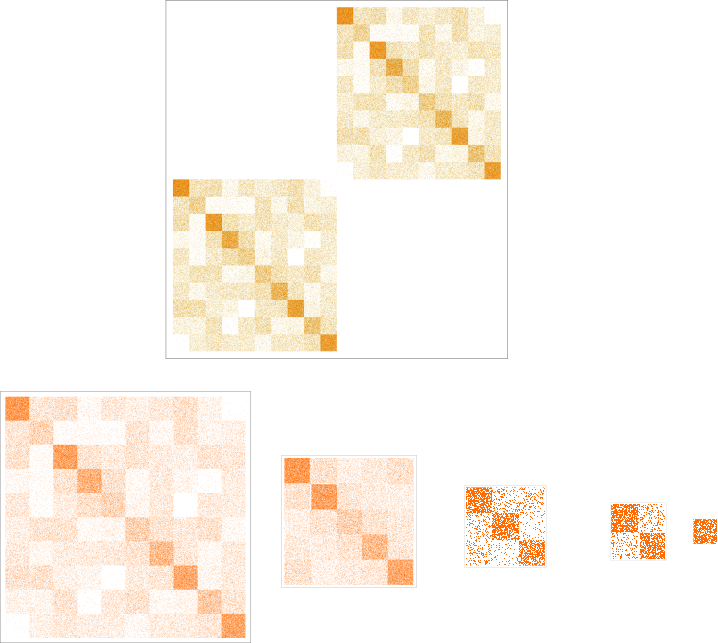}
\caption{Community detection with a classical annealer for a graph with $10$ communities and $2000$ nodes. Top: adjacency graph of the bipartite graph with communities. Lower row, from left to right, the stages of community elimination. The number of communities at different stages are $10$, $5$, $3$, $2$ and $1$. The last remaining community is the densest one.  }
\end{center}
\end{figure} 

\begin{figure}
\begin{center}
\includegraphics[scale=0.65
]{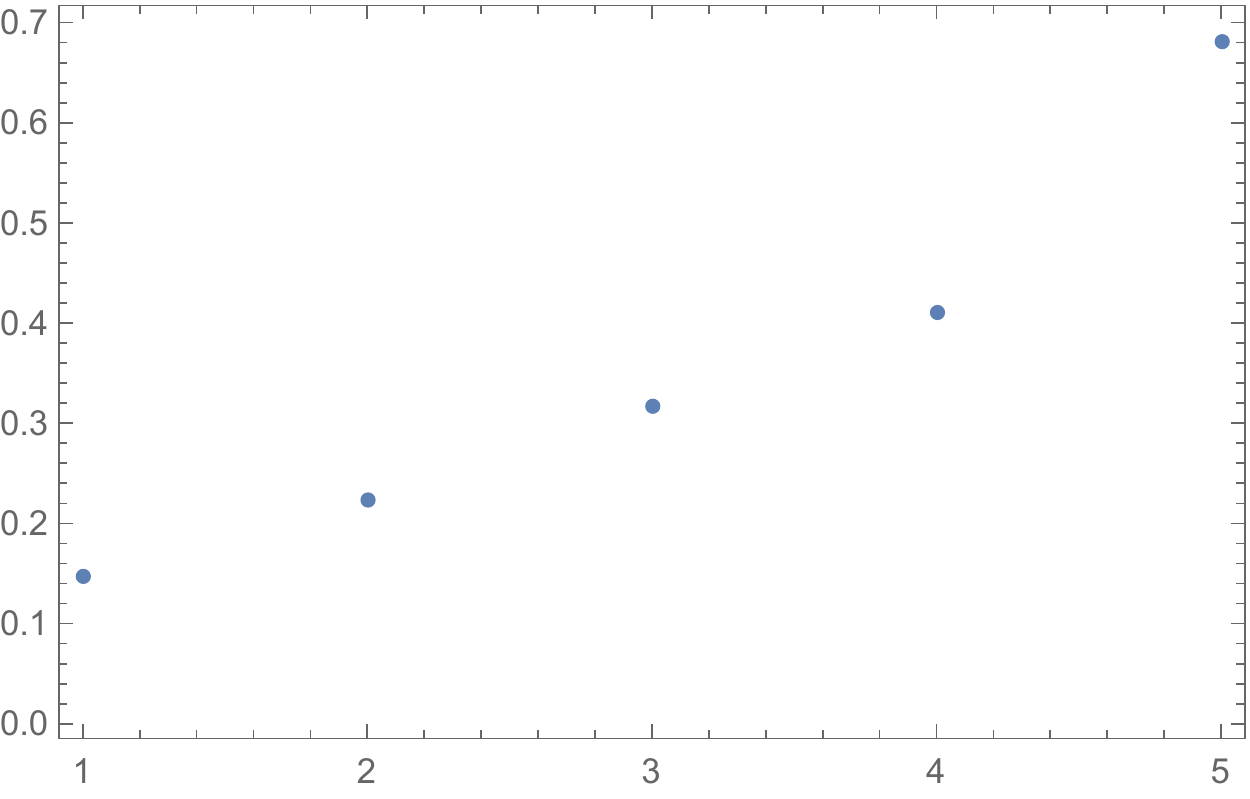}
\caption{Density of the subgraphs at stages $1-5$ of the community detection algorithm for the case shown in Fig. 11. The left-most dot indicates the density of the original graph $=0.146\cdots$. The density shows a steady growth until only one community is left, at step $5$. In our example, the last community is the densest of them, with a density of $0.680\cdots$. This circumstance could be used as a simple criteria of stopping the algorithm, since after one community is left a further increase of the density is expected to be very small.    }
\end{center}
\end{figure} 

\section{Algorithms}

We call our graph community detection algorithm as community panning. In this section we further scrutinize its details. The logical structure is given in the following Algorithm 1. 

The algorithm starts from a uniformly at random bi-partitioning of the input graph. The follows the steps of finding maximally dense subgraph in the sense of regularity check, as described in previous sections. Obviously there is a problem of stopping. We suggest to use the cost function divided by the product of sizes of bi-partitions. This can be called energy per node. We claim that such a function has minimum at the right step of the algorithm. In our experiments this suggestion works well, see Fig. 13.     
\begin{algorithm}
\caption{Community panning algorithm}\label{panning}
\begin{algorithmic}[1]
\Procedure{Find graph communities}{$G$}\Comment{Graph G}
\State Read adjacency matrix $A$ of $G$
\State Divide nodes ($V$) of $G$ in two random sets $V1$ and $V2$
\State Find non-zero block $B(V1,V2)$ of the adjacency matrix of the bipartite graph induced by $V1$ and $V2$ 
\State set $e1=0,\quad e2=0$
\While{$e1>e2$}\Comment{Stop at the minimum of energy per node}
\State $e1=e2$
\State Find non-zero block $B(V1,V2)$ of the adjacency matrix of the bipartite graph induced by $V1$ and $V2$, find link density of $B$, $d$ 
\State define qubo: $bqm=\sum_{i\in V1,j\in V2}s_i(d-B_{i,j})s_j$
\State call D-Wave to find $\argmin_{s_i,s_j}(bqm)=L$, $s_i\in\{0,1\}$
\State  $e=bqm(L)$
\State $n_1=\abs{V1},\quad n_2=\abs{V2}$, $L=L(s_1,\cdots, s_{n_1+n_2})$
\State $V1'=V1\quad V2'=V2$
\State update: $V1=\{i: 1\leq i\leq n_1, s_i=1\}$
\State update: $V2=\{i: n_1+1\leq i\leq n_1+n_2, s_i=1\}$
\State $n_1=\abs{V1},\quad n_2=\abs{V2}$
\State $e2=e/n_1/n_2$
\EndWhile\label{euclidendwhile}
\State \textbf{return} $V1'$ and $V2'$\Comment{$V1'\cup V2'$ is the list of nodes in the community}
\EndProcedure
\end{algorithmic}
\end{algorithm}

\begin{figure}
\begin{center}
\includegraphics[scale=0.4
]{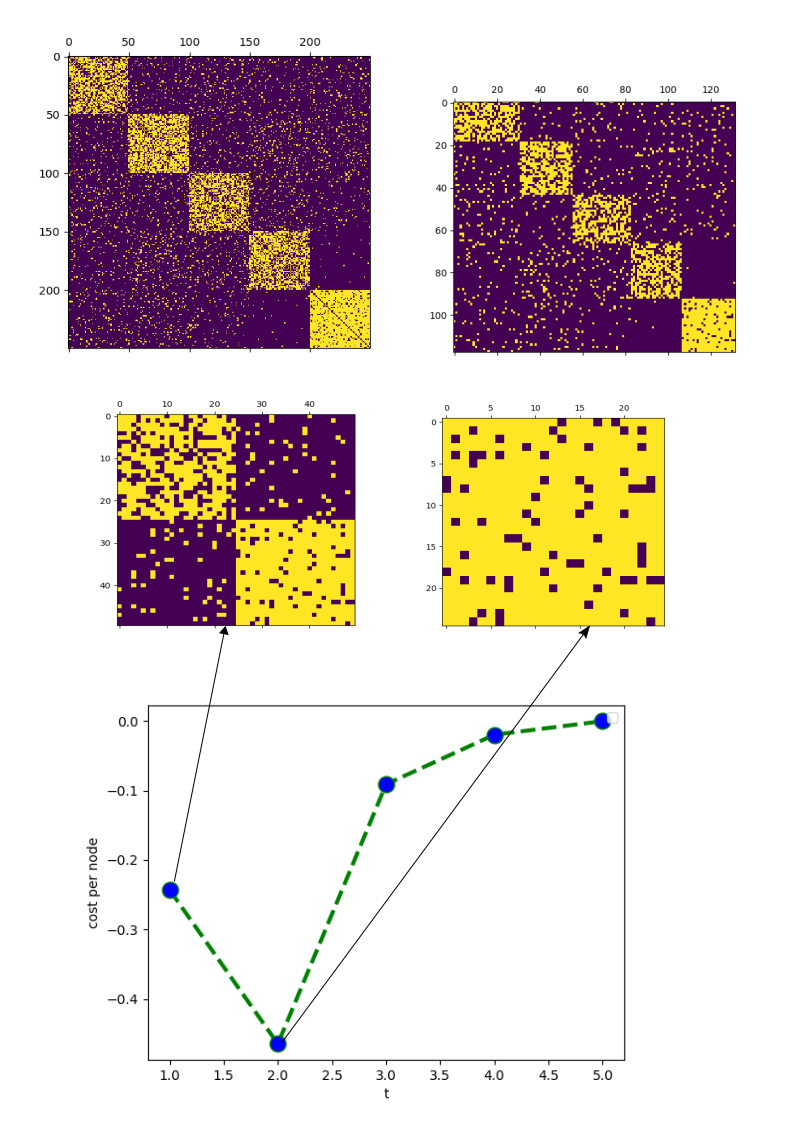}
\caption{The result of using community panning for a graph. Top-left adjacency matrix of the graph with communities. Yellow colour indicates value $1$  and dark colour $0$. The nodes are ordered in such a way that communities are apparent as diagonal blocks. From the adjacency matrix a random bipartite graph is chosen. Its non-zero block of the adjacency matrix is at the top right. First round of the algorithm chooses the corresponding subgraph with adjacency block shown at the middle left. Apparently it contains just two communities. Similarly the next step finds a community with the adjacency block in the middle right. It coincides perfectly with  one of the original communities, corresponding to the most yellow diagonal block in the whole graph in the top left figure. At the bottom is the plot of energy per node function at each stage. The step $t=2$ has minimum of this function. The subgraph corresponding to the minimum of the energy per node is the solution yielding one community. The other points in the energy plot show larger energy per node if the number of steps is too small or too large.       }
\end{center}
\end{figure} 

We made Python version of the corresponding algorithm available at the GitHub, \cite{compannning}. It can be used with D-Wave or without it implementing a version with classical annealing. The latter version is quite quick for moderate size graphs like the one in Fig. 13, with $250$ and five communities. 

Next we implemented an algorithm that loops the first algorithm to find all communities. In this case, there is a problem how to stop algorithm or in other words, how to decide when all communities are found. We suggest to use adjacency matrix visualisation or cost function plotting as way, see Fig. 14 for details. 
\begin{figure}
\begin{center}
\includegraphics[scale=0.4
]{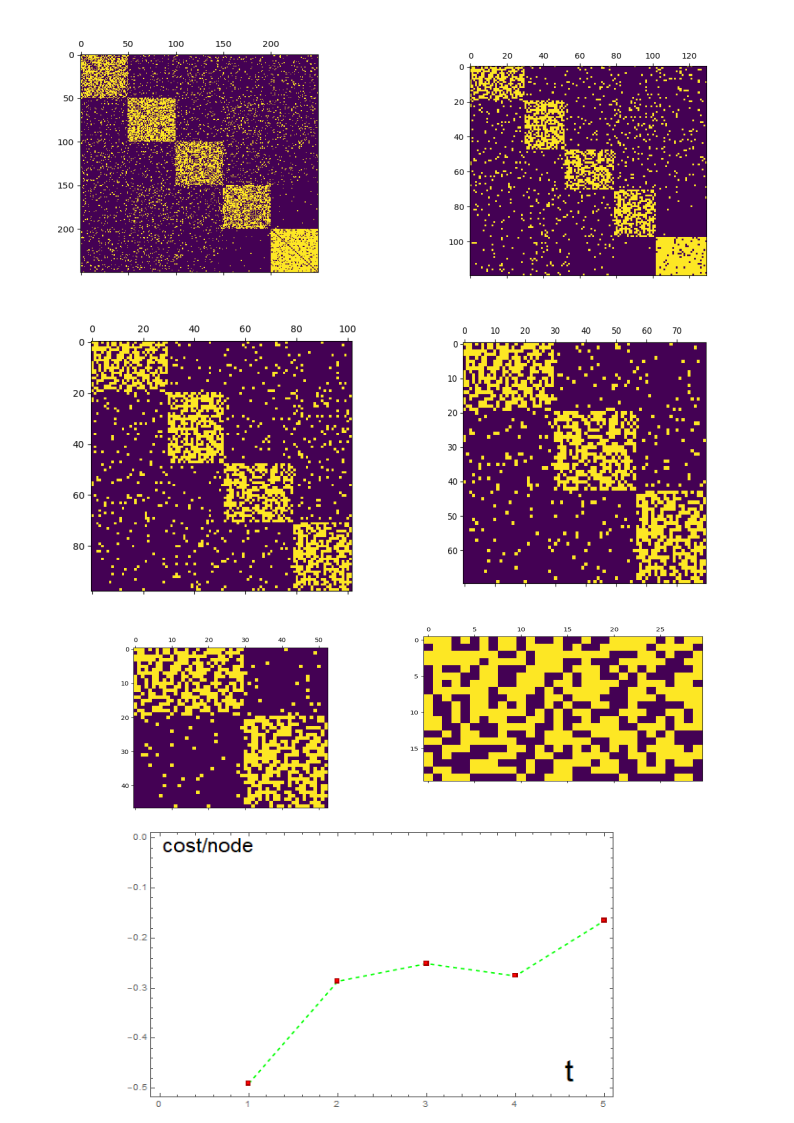}
\caption{The result of using community panning for a graph to find all communities. The graph is the same as in Fig.13. Top left matrix is the input adjacency matrix, the same as in Fig. 13. From top right to bottom right are the adjacency matrices of non-zero blocks at each stage. At the top right is the non-zero block of randomly bipartized  graph with five communities. At each stage one community is found and the corresponding nodes are deleted from the graph and the resulting graph is used as an input to the next stage. The bottom diagram shows energy per node of each solution. Parameter $t$ indicates number of iterations done. $t=5$, corresponds to stage when the algorithm is applied to a graph where only one community is left and no further are to be found. The cost function shows a gap, detection of which might be used as a stopping criteria. Another option could be use of adjacency matrix visualisation. In our plotting when we use the ordering of nodes according to community membership, it is clear that after four steps no further communities are to be found. In practise, this would mean that after each step, nodes are ordered in blocks and the whole adjacency matrix is plotted. The found communities should be visible as diagonal blocks similar to those in our plot.          }
\end{center}
\end{figure} 
\section{Conclusion}

Large graphs emerge from big data analysis and can pose serious computational problems. Szemer\'edi's Regularity Lemma (SRL) is a fundamental tool in large graph analysis and thus can become important also in the big data area.  

Quantum computing is a new emergent area in computation and can contribute in both areas.

We demonstrated connections between SRL, graph community detection and quantum annealing. SRL contains a very hard problem, regularity check, that could be solved on a quantum annealer, which we demonstrated using D-Wave quantum annealer. In case of graph community detection, we conjecture that quantum annealers of the future can produce high quality solutions for large scale problems.   

Research and business are already joining efforts in quantum computing and addressing big data. In $2019$, IT giant Google and German research center Forschungszentrum J\"ulich  announced a research partnership to develop quantum computing technology \cite{google}. 
\section*{Acknowledgment}
We would like to thank Dr. F\"ul\"op Bazs\'o for kindly sharing with us the cortical network data \cite{nepusz2008}. D-Wave Systems is acknowledged for providing us a free trial computation time in Leap cloud service.   This work was supported by BusinessFinland - Real-Time AI-Supported Ore Grade Evaluation for Automated Mining (RAGE) and Quantum Leap in Quantum Control (Qu2Co) projects.
\begin{appendix}
\subsection{Energy spectrum of quantum Hamiltonian}
D-Wave quantum computer uses large number of quantum bits or qbits. In this section our aim is to verify connection between quadratic optimization and the ground state of the corresponding quantum system in D-wave machine. We also review some basic concepts of quantum computation theory needed for this purpose.  
One qbit, a quantum version of a bit or a binary variable is described by a vector  $\ket{Z}\in \mathbb{C}^2$, which is the complex vector space of dimension 2. $\ket{Z}$ are referred as ket-vectors.   
 
Elements  of dual vector space are called bra-vectors and are denoted as  $\bra{Z}\in \mathbb{C}^{2*}$, which forms a complex vector space of dimension 2 

 $\mathbb{C}^{2*}$ is just space of all row vectors with two complex coordinates: $(z_1,z_0)$
 
 Notations:
 \begin{eqnarray*}
 e_1=\left(
\begin{array}{cc}
 1\\
 0
\end{array}
\right):=\ket{1},e_0=\left(
\begin{array}{cc}
 0\\
 1
\end{array}
\right):=\ket{0},\\ e_1^T=(1,0)=\bra{1}, e_0^T=(0,1):=\bra{0},
 \end{eqnarray*}
 where, $T$ stands for matrix transposition. 
 
 Any vector in $\mathbb{C}^2$ or $\mathbb{C}^{2*}$ can be written:
  $Z:= \ket{Z}= z_1\ket{1}+z_0\ket{0}$, $\bra{Z}=\ket{Z}^\dagger=\bra{1}z_1^*+\bra{0}z_0^*$,
  where $\dagger$ stands for transposition and complex conjugate (Hermite transpose).
  Inner product is understood as a matrix multiplication:
 \begin{eqnarray*}
 \bra{Z}\ket{Z'}=(\bra{1}z_1^*+\bra{0}z_0^*)(z_1'\ket{1}+z_0'\ket{0})=\\
 z_1^*z_1'\bra{1}\ket{1}+z_1^*z_0'\bra{1}\ket{0}+z_0^*z_1'\bra{0}\ket{1}+z_0^*z_0'\bra{0}\ket{0}=\\
 z_1^*z_1'+z_0^* z_0',
 \end{eqnarray*}
 where we used orthogonality  conditions $\bra{i}\ket{j}=\delta_{i,j}$, with $i,j\in\{0,1\}$ and $\delta_{i,j}=1$ if $i=j$ and zero otherwise. These conditions follow directly from our definitions. 

 Classical register of length $n$ is a tuple of binary variables:
    $R=(b_1,b_2,\cdots,b_n)$, $b_i\in\{0,1\}$
    
    Quantum register of length $n$ is a system of $n$ qbits. As in one qbit case, states in which all qbits have some particular value form a basis and are denoted as:
    $$
    \ket{R}=\ket{b1}\ket{b_2}\cdots \ket{b_n}=\ket{b_1b_2\dots b_n},
    $$
    We assume that the carriers of qbit states are not identical entities, which is usually the case. In the opposite case the state vectors would have additional bosonic or fermionic symmetry property upon permutations. 
    
    If $n$ qbits are in any state $\ket{R}$ result of measuring of qbits is always register $R$. Such states we call the product states. Qbits can be in any linear combination of such product states with arbitrary complex numbers, which add up to on in square modules,  as coefficients:
    $$
    \ket{\Psi}= \sum_{b_1,\cdots,b_n}z_{b_1,\cdots,b_n}\ket{b_1\cdots b_n}, \sum_{b_1,\cdots,b_n}\abs{z_{b_1,\cdots,b_n}}^2=1
    $$
    
    State space of quantum register with $n$ qbits ($H_n$) is called tensor product of n spaces $\mathbb{C}^2$:
  $$
\ket{\Psi}\in H_n=\mathbb{C}^2\otimes\mathbb{C}^2\otimes\cdots \otimes\mathbb{C}^2=(\mathbb{C}^2)^{\otimes n}.
  $$
  Defining the inner product in $H_n$ as:
  $$
  \bra{\Psi}\ket{\Psi'}= \sum_{b_1,\cdots,b_n\in\{0,1\}}z_{b_1,\cdots,b_n}^* z_{b_1,\cdots,b_n}',
  $$
  $H_n$ is turned into a Hilbert space of dimension $2^n$. 
  
  The interpretation of quantum register state is probabilistic. When qbits are measured, probability of an outcome $(b_1,\cdots, b_n)$ in state $\ket{\Psi}$ is: 
  \begin{axiom}
  $$
  \mathbb{P}((b_1,\cdots, b_n)\ket{\Psi})=\abs{z_{b_1,\cdots,b_n}}^2.
  $$
  \end{axiom}
 A general state in $H_n$ is called entangled state. An entangled state is a superposition of product states. The product states, where each qbit has a definite state, is only a tiny fraction of the $H_n$: 
\begin{proposition}
A state in $H_n$ is described by a $2^{n+1}-2$ dimensional manifold while product states by a $2n$ dimensional real manifold.
\end{proposition}
\begin{proof}
By dimension of a manifold we mean just number of real parameters needed to uniquely define a quantum state. A general entangled state in $H_n$ has $2^n$ complex coordinates, each coordinate needs two real parameters. This gives $2\times 2^n=2^{n+1}$ real parameters. The normalization condition reduces number of parameters by one. The states $\ket{\psi}$ and $e^{i\phi}\ket{\psi}$, describe the same quantum states, sharing a common phase factor $e^{i\phi}$. This reduces number of parameters by one. The result is as claimed $2^{n+1}-2$. A product state can be written as: $(c_1\ket{0}+c'_1\ket{1})\otimes((c_2\ket{0}+c'_2\ket{1})\otimes\cdots \otimes (c_n\ket{0}+c'_n\ket{1})$, each factor $(c_i\ket{0}+c'_i\ket{1})$ has two complex parameters and one normalization condition $\abs{c_i}^2+\abs{c'_i}^2=1$, this result in $3$ parameters.  The common phase factor reduces the number of parameters to $2$. The whole state is thus described by $2n$ real parameters.
\end{proof}
Take $n=333$, then a state of $H_n$ is described by more than $10^{100}$ real parameters. Such a manifold is impossible to 'digitize', at least in our universe! 

Operators in $n$-qbit Hilbert space, $H_n=(\mathbb{C}^2)^{\otimes n}$, are in general $2^n\times2^n$ matrices.  One way of defining such operators is a tensor product of operators. Let $X_1$ and $X_2$ be operators acting in space of the first - and the second qbit. The tensor product of these operators acts on product states according to rule:
$$
(X_1\otimes X_2)(\ket{\psi_1}\otimes \ket{\psi_2}):= X_1\ket{\psi_1}\otimes X_2\ket{\psi_2}.
$$
The $X_1\otimes X_2$ is defined in the whole $H_2$ by linearity. Generalization to a case of arbitrary number of factors $n$ is done similarly. 

Quantum annealing system in D-Wave architecture has Hamiltonian:
$$
H_{DW}=\sum_i h_i\sigma_{zi}+\sum_{1\leq i< j\leq n}J_{ij}\sigma_{zi}\otimes \sigma_{zj},
$$
in which $\sigma_{zi}$ is a Pauli matrix $\sigma_z$ acting on qbit $i$, $h_i$ and $J_{i,j}$ are adjustable real parameters. This is a short hand notation, unit matrix $I$ is assumed in the tensor product not occupied by $\sigma_z$, for instance $h_1\sigma_{z1}=h_1\sigma_{z1}\otimes I_2\otimes\cdots\otimes I_n$ and so on.
     \begin{proposition}
    The spectral decomposition of $H_{DW}$ is
    $$
    H_{DW}= \sum_{b_i\in\{0,1\}, i=1,2,\cdots, n} \ket{b_1,b_2,\cdots,b_n}\bra{b_1,b_2,\cdots,b_n}E_I,
    $$
    In which the Ising energy $E_I(b_1,\cdots,b_n)= \sum_ih_is_i+\sum_{1\leq i<j\leq n}J_{i,j}s_is_j$ and in which $s_i=2b_i-1$.
    \end{proposition}
 \begin{proof}
 $$
  \sigma_z\ket{1}=\left(
\begin{array}{cc}
 1,\quad 0\\
 0,-1
\end{array}
\right) 
\left(
\begin{array}{cc}
 1\\
 0
\end{array}
\right)
=\left(
\begin{array}{cc}
 1\\
 0
\end{array}
 \right),
$$
$$
\sigma_z\ket{0}=\left(
\begin{array}{cc}
 1,\quad 0\\
 0,-1
\end{array}
\right) 
\left(
\begin{array}{cc}
 0\\
 1
\end{array}
\right)
=-\left(
\begin{array}{cc}
 0\\
 1
\end{array}
\right)
$$
As a result $\sigma_z\ket{b}=(2b-1)\ket{b}$, $b=0,1$. That is why $H_{DW}\ket{b_1,b_2,\cdots, b_n}=E_I(b_1,\cdots,b_n)\ket{b_1,b_2,\cdots, b_n}$, indeed, for instance  $\sigma_{zi}\otimes\sigma_{zj}\ket{b_1,b_2,\cdots, b_n}=(2b_i-1)(2b_j-1)\ket{b_1,b_2,\cdots, b_n}$. Vectors $\ket{b_1,b_2,\cdots, b_n}$, with all possible combinations of $b_i$ form an orthonormal basis in $H_n$, and because all of them are eigenvectors of $H_{DW}$, with corresponding eigenvaluess $E_I$, and the spectral decomposition follows.  
\end{proof}   

Corollary:
Ground-state of the D-Wave system corresponds to $\min E_I$. Eigenstates are product states - no entanglement. 

    \subsection{Regularity check using quantum existence algorithm}
    The regularity check (RC) of a bipartite graph is computationally hard problem requiring, in the worst cases, exponential time in number of nodes. Quantum search using so called Grover's algorithm (see e.g. \cite{portugal}) could improve the time needed for the RC. The time is still exponential but with smaller base of the exponential function. This corresponds to the famous speed-up from $N$ to $\sqrt{N}$, corresponding to exhaustive search of $N$ items with classical versus Grover's search.
    
    Grover's algorithms finds solution to the following problem: let among $N=2^m$ items, encoded as integers $M:=\{0,\cdots,N-1\}$, be exactly one marked element $x_0\in M$; find $x_0$. It is assumed that there is a function $f: M\rightarrow \{0,1\}$ such that $f(x_0)=1$ and otherwise $f(x)=0$ for all other $x\in M$. It is assumed that $f$ can be computed quickly. In other words if the solution is found, it can be easily verified as such.
    
    First step is to map items to a register of $m$ qbits. Each number $x\in M$ is written as a binary tuple of $m$ bits: $(x_1,\cdots,x_m)$, $x_i\in\{0,1\}$. Then such a number is mapped to the quantum state of $m$ qbits: $\ket{x}=\ket{x_1\cdots x_m}$. 
    
    The quantum oracle is an unitary operator acting according to the rule: $R_f\ket{x_0}=-\ket{x_0}$ and $R_f\ket{x}=\ket{x}$, $\forall x\neq x_0$. Superposition of all states corresponding to $M$ is denoted as $\ket{D}=\frac{1}{\sqrt{N}}\sum_{x\in M}\ket{x}$. $R_D:=2\ket{D}\bra{D}-I$ is another unitary operator needed. Product of these two operators is unitary operator $U:=R_DR_f$, which is the main transformation in Grover's algorithm. The process starts from state $\ket{D}$; after $t\in\mathbb{N}$ steps the system is in the state:
    
    $$
    \ket{\psi_t}:=U^t\ket{D}, \quad t=1,2,\cdots.
    $$
    If $t_0=\floor{\frac{\pi}{4}\sqrt{N}}$, then with probability $p_0>1-\frac{1}{N}$, the state $\ket{\psi_{t_0}}$ is $\ket{x_0}$, and the binary representation of the solution $x_0$ is found by measuring the state of the quantum register $\ket{\psi_{t_0}}$.
    
    An analysis shows that $U^t\ket{D}$ is a rotation by an angle $t\theta$ in two dimensional plane spanned by vectors $\ket{x_0}$ and $\ket{D}$ in which $\theta\approx \frac{2}{\sqrt{N}}$ is a constant. In this plane $U$ has representation:
     $$
  U=\left(
\begin{array}{cc}
 \cos(\theta ), -\sin(\theta )\\
 \sin(\theta ),\quad\cos(\theta )
\end{array}
\right). 
$$
 Eigenvalues of this matrix are $e^{i\theta}$  and $e^{-i\theta}$. The corresponding eigenvectors are:
 $$
 \ket{\alpha_1}=\frac{\ket{0}-i\ket{1}}{\sqrt{2}},\quad \ket{\alpha_2}=\frac{\ket{0}+i\ket{1}}{\sqrt{2}}
 $$
  in which $\ket{1}$  corresponds to the state $\ket{x_0}$ and $\ket{0}$ corresponds to the vector orthogonal to $\ket{x_0}$ in the two dimensional plane we described above. In these notations the initial vector $\ket{D}$ can be written as:
  $$
  \ket{D}=\frac{\sqrt{N-1}}{\sqrt{N}}\ket{0}+\frac{1}{\sqrt{N}}\ket{1}=\cos(\frac{\theta}{2})\ket{0}+\sin(\frac{\theta}{2})\ket{1}.
  $$
\end{appendix}
As a result we have:
$$
U^t\ket{D}= \cos(\theta t+\frac{\theta}{2})\ket{0}+ \sin(\theta t+\frac{\theta}{2})\ket{1},
$$
which indicates that square of amplitude of the $\ket{1}$, which is $\sin^2(\theta t+\frac{\theta}{2})$, at step $t=t_0$ is very close to one. This means that when the register is read/measured, the result is $\ket{x_0}$ with probability very close to one, assuming large $N$. This also indicates that $t_0$ should be known. For $t>t_0$, the amplitude of $\ket{x_0}$ starts to decrease. 

In case there are $\abs{M}\geq 1$ solutions to equation $f(x)=1$, similar algorithm works. However, the $\theta$-parameter changes to $\theta= 2\sqrt{\frac{\abs{M}}{N}}$, and similarly number of steps changes. As a result, in order to use Grover's algorithm, the number of solutions need to be known beforehand. For details, see e.g. \cite{portugal}.

It appears that there exist quantum algorithms for evaluating number of solutions without actually finding the solutions, \cite{imre}. One of them is based on estimating eigenvalue of an unitary operator, which is of the form $e^{i\theta}$, $\theta\in [0,2\pi)$. Indeed if we can find eigenvalue of $U$ of the Grovers's operator: $U\ket{\alpha_1}=e^{i\theta}\ket{\alpha_1}$, we can find out number of solutions, since $\theta= 2\sqrt{\frac{\abs{M}}{N}}$. This covers the case in which there is no solution: $U=I$, and $\theta=0$. This problem is called quantum phase estimation problem. The task of finding out whether $\abs{M}=0$ or $\abs{M}>0$, is called quantum existence problem \cite{imre}. 

Quantum phase estimation, see \cite{cleve}, is based on quantum Fourier transformation ($QFT$). $QFT$ is a linear mapping between $m$ qbit states:
$$
QFT(\ket{a})= \frac{1}{\sqrt{N}}\sum_{k=0}^{N-1}e^{\frac{2\pi ia k}{N}}\ket{k}. \quad  a\in\{0,1,\cdots, N-1\}
$$
the inverse mapping is:
$$
QFT^{-1}(\ket{a})=\frac{1}{\sqrt{N}}\sum_{k=0}^{N-1}e^{-\frac{2\pi ia k}{N}}\ket{k}.
$$
For any superposition of states $QFT$ and $QFT^{-1}$ are defined by linearity, say, $QFT(\sum_{i}c_i\ket{i})=\sum_ic_iQFT(\ket{i})$.

In case of $m$ qbits we have integers $a: 0\leq a\leq 2^{m-1}$, and we can write $a$ in binary basis: $a=2^{m-1}a_m+2^{m-1}a_2+\cdots+2^0a_1$, where coefficients $a_i$ are binary numbers. Then, by definition, we write:
$$
\ket{a}:=\ket{a_m}\otimes\ket{a_{m-1}}\otimes\cdots\otimes \ket{a_1}:=\ket{a_m\cdots a_1}.
$$
It appears that $QFT$ can be written as:
\begin{eqnarray*}
QFT(\ket{a})=(\ket{0}+e^{2\pi i(0.a_m)})(\ket{0}+e^{2\pi i(0.a_ma_{m-1})}\ket{1})\\\cdots (\ket{0}+e^{2\pi i(0.a_m\cdots a_1)}\ket{1}).
\end{eqnarray*}
Let assume that unitary operator $U$ has eigenstate $\ket{u}$, with eigenvalue $e^{i\theta}=e^{2\pi i0.a_ma_{m-1}\cdots a_1}$ ( $m $-binary-digit eigenvalue).

An operator $ C\_U$, called the controlled $U$. It acts on vectors like $\frac{1}{\sqrt{2}}(\ket{0}+\ket{1})\otimes \ket{u}$. $C\_U$ acts as $U$ on the $\ket{u}$ if the value  of the first qbit is 1, otherwise it is unit operator, for instance:
\begin{eqnarray*}
C\_U\frac{1}{\sqrt{2}}(\ket{0}+\ket{1})\ket{u}=\frac{1}{\sqrt{2}}(\ket{0}\ket{u}+\ket{1} U\ket{u})=\\\frac{1}{\sqrt{2}}(\ket{0}\ket{u}+e^{i\theta}\ket{1}\ket{u})=\frac{1}{\sqrt{2}}(\ket{0}+e^{i\theta}\ket{1}) \ket{u}.
\end{eqnarray*}
The phase estimation uses the following scheme. First register has $m$ qbits, all in states $\frac{1}{\sqrt{2}}(\ket{0}+\ket{1})$, second register is in the state $\ket{u}$, the eigenstate of $U$, phase of which we wish to find. Then controlled $U$ operators are applied in sequence: $C_j\_U^{2^j}$, $j=0,1,\cdots, m-1$, where controlled $U$ with index $j$, uses $j$th qbit from the bottom of the first register as a control qbit. It easy to check that the state of the register after these operations equals to (omitting the normalization coefficient):
\begin{eqnarray*}
(\ket{0}+e^{i\theta2^{m-1}}\ket{1})(\ket{0}+e^{i\theta2^{m-2}}\ket{1})\cdots\\(\ket{0}+e^{i\theta2^0}\ket{1})\ket{u}=
\end{eqnarray*}
\begin{eqnarray*}
(\ket{0}+e^{2\pi i0.a_m}\ket{1})(\ket{0}+e^{2\pi i0.a_ma_{m-1}}\ket{1})\cdots
\end{eqnarray*}
\begin{eqnarray*}
(\ket{0}+e^{2\pi i0.a_ma_{m-1}\cdots a_1}\ket{1})\ket{u}.
\end{eqnarray*}
As a result this is just $QFT(\ket{a_ma_{m-1}\cdots a_1})\ket{u}$. By applying the $QFT^{-1}$ to the first register we get the state $\ket{a_ma_{m-1}\cdots a_1}\ket{u}$, and by measuring the first registers, all digits of the phase can be read. 

In the case when $\theta$ has more significant digits, the above scheme still provides an approximation for the angle, see \cite{cleve}, which is very accurate for large enough $m$.

In the case of the Grovers's algorithm, we use its unitary operator as $U$ in the phase estimation algorithm. One issue is that the initial state $\ket{d}$ is a linear combination of eigenvectors with corresponding phases $\theta$ and $-\theta$. The phase estimation algorithm results in either $\theta$ or $2\pi-\theta$, and in both cases it is possible to find out estimate of $\theta$. Omitting normalization coefficients of the state, we have $\ket{d}=\ket{\alpha_1}+\ket{\alpha_2}$, with $U\ket{\alpha_1}=e^{i\theta}\ket{\alpha_1}$ and $U\ket{\alpha_2}=e^{-i\theta}\ket{\alpha_2}$. To get $m$ digit approximation of $\theta$, we start from the state $(\ket{0}+\ket{1})^{\otimes m}(\ket{\alpha_1}+\ket{\alpha_2})= (\ket{0}+\ket{1})^{\otimes m}\ket{\alpha_1}+ (\ket{0}+\ket{1})^{\otimes m}\ket{\alpha_2}$. Applying the sequence of linear operators, $C_j\_U^{2^j}$, in the phase estimation algorithm, the result is $QFT(\ket{\theta})\ket{\alpha_1}+QFT(\ket{2\pi-\theta})\ket{\alpha_2)}$, in which $\ket{\theta}$ corresponds to $m$-digits approximation of the angle $\theta$. Applying $QFT^{-1}$ and measuring the first register will give bit string of $\theta$ or $2\pi-\theta$ at random. 

Next we show how the Grover's algorithm and the quantum phase estimation can be used to solve the RC problem. Use the same notation as in the main text in the Section II. We have a bipartite graph $G(V_1,V_2,d)$, with $n_1+n_2$ nodes. Define:
$$
F(X,Y):=\abs{\abs{ X}\abs{Y}(d(A,B) - d(X,Y))}, X\subset V_1, Y\subset V_2
$$
$G$ is $\epsilon$-regular iff
$$
\forall X\subset V_1, Y\subset V_2: F(X,Y)\leq \epsilon n_1n_2
$$
We encode $X$ and $Y$ using $n_1+n_2$ bits. This description can be mapped to $m$-qbit states, in which bit value $1$ indicates that the corresponding node belongs to $X$ or $Y$. 

We define the oracle function as:
$f(X,Y)=1$ iff $F(X,Y)>\epsilon n_1n_2$, otherwise $f(X,Y)=0$. The function $f$ is easily computable for any argument.
The RC can be formulated as:

Graph $G$ is $\epsilon$-regular only if $f(X,Y)=1$ has no solution. 
\begin{proposition}
Regularity check problem can be resolved in $\sim\sqrt{2^{n_1+n_2}}$ steps of using a quantum gate computer
\end{proposition}
\begin{proof}(A sketch)
The standard operators and quantum states to use Grover's algorithm to find solutions of $f(X,Y)=1$, are obvious in our formulation of RC. In RC we only need to know whether $M>0$ or not. That is why the quantum phase estimation is appropriate. From Grover's algorithm it is known that if $M>0$, the phase has order of magnitude $\sim 1/\sqrt{N}$, where $N:=2^{n_1+n_2}$. That is why, we need approximation of $\theta$ that has of the order of $m=\log_2(\sqrt{N)}=\frac{1}{2}(n_1+n_2)$ binary digits. Using the phase estimation algorithm we need to apply controlled Grover's transform  $2^{m-1}+2^{m-2}+\cdots + 1=2^m=\sqrt{2^{n_1+n_2}}$ times. This comes from the controlled Grover's transforms $C_j\_U^{2^j}$ in which $U$ is applied $2^j$ times.
\end{proof}
%
%
{}
\end{document}